\documentclass[11pt]{article}

\usepackage{authblk}
\usepackage{hyperref}
 
\usepackage{microtype}
\usepackage{algorithm}
\usepackage{algorithmicx}
\usepackage{tikz}
\usepackage{algpseudocode}
\usepackage{xspace}
\usepackage{mdframed}
\usepackage{stmaryrd}

\usepackage{amsmath,amssymb,amsthm,amsfonts,latexsym,graphicx}
\usepackage{fullpage,color}
\usepackage{enumerate}
\usepackage{url,hyperref}
\usepackage{comment}

\algtext*{EndWhile}
\algtext*{EndIf}
\algtext*{EndFor}
\algtext*{EndFunction}

\usetikzlibrary{snakes}
\theoremstyle{plain}

\newtheorem{lemma}{Lemma}
\newtheorem{theorem}[lemma]{Theorem}

\newtheorem{observation}[lemma]{Observation}

{\medskip}

\title{Faster Algorithms for All Pairs Non-decreasing Paths Problem}
\author[1]{Ran Duan \thanks{duanran@mail.tsinghua.edu.cn}}
\author[1]{Ce Jin \thanks{jinc16@mails.tsinghua.edu.cn}}
\author[1]{Hongxun Wu \thanks{wuhx18@mails.tsinghua.edu.cn}}

\affil[1]{Institute for Interdisciplinary Information Sciences, Tsinghua University}

\date{}

\begin{document}

\maketitle

\begin{abstract}
In this paper, we present an improved algorithm for the All Pairs Non-decreasing Paths (APNP) problem on weighted simple digraphs, which has running time $\tilde{O}(n^{\frac{3 + \omega}{2}}) = \tilde{O}(n^{2.686})$. Here $n$ is the number of vertices, and $\omega < 2.373$ is the exponent of time complexity of fast matrix multiplication [Williams 2012, Le Gall 2014]. This matches the current best upper bound for $(\max, \min)$-matrix product [Duan, Pettie 2009] which is reducible to APNP. Thus, further improvement for APNP will imply a faster algorithm for $(\max, \min)$-matrix product. The previous best upper bound for APNP on weighted digraphs was $\tilde{O}(n^{\frac{1}{2}(3 + \frac{3 - \omega}{\omega + 1} + \omega)}) = \tilde{O}(n^{2.78})$ [Duan, Gu, Zhang 2018]. We also show an $\tilde{O}(n^2)$ time algorithm for APNP in undirected graphs which also reaches optimal within logarithmic factors.
\end{abstract}

\section{Introduction} \label{intro}

Given a directed or undirected graph $G=(V,E)$ with arbitrary real edge weights, a \textit{non-decreasing path} is a path on which edge weights form a non-decreasing sequence~\cite{minty1958a}. We define the weight of a non-decreasing path to be the weight of its last edge, which we want to minimize. The motivation of this definition comes from train scheduling \cite{williams2010nondecreasing}.
Suppose each train station is mapped to a vertex of a directed graph, and a train from station $v_1$ to station $v_2$ scheduled at time $w$ is mapped to a directed edge $(v_1, v_2)$ with weight $w$.
If we neglect the time trains spent on their way, a trip from $s$ to $t$ is just a non-decreasing path from $s$ to $t$ in the constructed graph, and the earliest time arriving at $t$ equals the minimum weight of such non-decreasing path. If we consider the train starts from $v_1$ at time $w_1$ and arrives at $v_2$ at time $w_2$, we can add a vertex $u$ and two edges $(v_1,u), (u,v_2)$, then gives edge weights $w_1,w_2$ on them, respectively.

The \textit{Single Source Non-decreasing Paths} (SSNP) problem, first studied by Minty \cite{minty1958a} in 1958, is to find the minimum weight non-decreasing path from a given source $s$ to all $t\in V$. The first linear time algorithm for SSNP problem in RAM model was given by Vassilevska Williams~\cite{williams2010nondecreasing}. She also gave an $O(m \log \log n)$ time algorithm in the standard addition-comparison model. ($m$ is the number of edges.)

The \textit{All Pairs Non-decreasing Paths} (APNP) problem is the all-pairs version of the above problem, asking for the minimum weight non-decreasing path between \textit{all} pairs of vertices in the graph.
Vassilevska Williams~\cite{williams2010nondecreasing} gave the first truly sub-cubic time algorithm of APNP.
The original time complexity of the algorithm was $\tilde{O}(n^{(15 + \omega)/6})$, obtained by using an $\tilde O(n^{2+\omega/3})$-time $(\min, \le)$-matrix product algorithm~\cite{vassilevska2009all} as a subroutine. Here $\omega < 2.373$ is the exponent of time complexity of fast matrix multiplication~\cite{Coppersmith:1987:MMV:28395.28396,williams2012multiplying,le2014powers}, and the $(\min,\le)$-product of two matrices $A,B$ is defined as $C_{i,j} =\min_{k} \{B_{k,j}: A_{i,k} \le B_{k,j}\}$.
It can be improved to $\tilde O(n^{(9+\omega)/4})$ time by using a faster $(\min, \le)$-matrix product algorithm in $\tilde O(n^{(3+\omega)/2})$ time implied in~\cite{duan2009fast}.
Recently, Duan et al. \cite{duan2018improved} obtained a faster algorithm for APNP in $\tilde{O}(n^{\frac{1}{2}(3 + \frac{3 - \omega}{\omega + 1} + \omega)}) = \tilde{O}(n^{2.78})$, by using a simple balancing technique introduced in~\cite{duan2009fast}.
We also adapt this technique in our algorithm. 

Computing APNP is at least as hard as the \textit{All Pair Bottleneck Paths} (APBP) problem~\cite{williams2010nondecreasing}, which asks for the maximum bottleneck path between every pair of vertices, where the bottleneck of a path is defined as the minimum weight (capacity) among all edges on this path. The current fastest algorithm for APBP runs in $\tilde O(n^{(3+\omega)/2}) = \tilde{O}(n^{2.686})$ time \cite{duan2009fast}. Our algorithm for APNP matches this running time, so any further improvement on APNP will imply a faster algorithm for APBP as well.

The \textit{vertex-weighted} APNP problem on directed graphs, a restricted version of APNP, is computationally equivalent to the
problem of \textit{Maximum Witness for Boolean Matrix Multiplication} (MWBMM) \cite{williams2010nondecreasing}. 
An algorithm of $O(n^{2+\mu})$
time for MWBMM was given by Czumaj et~al.~\cite{czumaj2007faster}, where $\mu$ satisfies the equation
$\omega(1, \mu, 1) = 1 + 2\mu$ and $\omega(1, \mu, 1)$ is the exponent of multiplying an $n \times  n^\mu$ matrix with
an $n^\mu \times n$ matrix. Currently, the best bounds on rectangular matrix multiplication by Le Gall and Urrutia
\cite{gall2018improved} imply that $\mu < 0.5286$. 



\subsection{Our results}
In this paper we describe a faster algorithm for directed APNP problem running in $\tilde O(n^{(3 + \omega)/2})$ time, which reaches optimal if the algorithm for APBP cannot be improved. 
\begin{theorem}\label{mainthmdirected}
The all pairs non-decreasing paths (APNP) problem on directed simple graphs can be solved in $\tilde{O}(n^{(3 + \omega)/2})$ time.
\end{theorem}
As in Dijkstra search~\cite{dijkstra1959note} we can maintain a priority queue of current non-decreasing paths we have found, then the minimum unvisited one is the optimal paths between its endpoints. Every time we visit an optimal path, we ``relax'' all edges following that path. In~\cite{duan2018improved} they partition the edge set into some parts by increasing order. For low-degree vertices in one part, trivially relax all of its outgoing edges, while for high-degree ones, use matrix multiplication to accelerate. Our algorithm adapt this idea, but we recursively divide the edges to $O(\log n)$ levels. In order the optimize the running time, a careful analysis of matrix multiplication is needed, and we need new ideas to use fast matrix multiplication to ``predetermine'' some of the useless edges from high-degree vertices.

We also give an $\tilde O(n^2)$ time algorithm for undirected APNP, which reaches optimal within logarithmic factors.
\begin{theorem}
\label{mainthmundirected}
The all pairs non-decreasing paths (APNP) problem on undirected simple graphs can be solved in $\tilde{O}(n^2)$ time.
\end{theorem}

\subsection{Organization of this paper}
Section~\ref{prelim},~\ref{BT},~\ref{AA} will discuss our APNP algorithm for directed graphs. In Section~\ref{prelim}, we introduce some terminologies used throughout this paper, and discuss how to recursively divide the edges. Then Section~\ref{BT} and~\ref{AA} will discuss our APNP algorithm for directed graphs in detail, where Section~\ref{BT} explains the main techniques used, and Section~\ref{AA} describes our whole algorithm and its analysis using procedures in Section~\ref{BT} as subroutines. The algorithm for undirected graphs is included in Section~\ref{undirectedsection}.
The graphs we considered in these sections have distinct edge weights, so generalization of our algorithm to equal edge weights is discussed in Appendix~\ref{equal-weight}. 

\section{Preliminaries}\label{prelim}

\subsection{Basic definitions}
\paragraph*{APNP problem}
Let $G = (V,E)$ be a directed simple graph with edge weight $w(e)$ for each edge $e \in E$. We denote $n = |V|$ and $m = |E| = O(n^2)$. 

A path is a sequence of edges $e_1, e_2, \dots, e_l$. A non-decreasing path is a path satisfying $\forall 1\leq i\leq l-1$, $w(e_i) \leq w(e_{i+1})$, and the weight of this non-decreasing path is defined to be $w(e_l)$, the weight of the last edge.
All pairs non-decreasing paths problem asks to determine the minimum weight non-decreasing path between every pair of vertices. 
Let $OPT(i,j)$ denote the optimal (minimum) non-decreasing path between $i$ and $j$. 
In contrast, during the algorithm, we use $d(i,j)$ to denote the current minimum answer that our algorithm has found so far. For convenience, $OPT(i,j)$ and $d(i,j)$ also denote the weight of that path, and $w(j,k)$ denotes the weight of the edge $(j,k)$. 

\paragraph*{Notation for String and Subgraph}

Without loss of generality\footnote{We include the proof in appendix \ref{equal-weight} for completeness.}, we assume all edges have distinct edge weights ranged from $0$ to $2^b-1$, where  $b = \lceil \log_2 |E|\rceil \le 2\log_2 n + O(1)$, so $2^b=O(n^2)$. Then every edge weight $w(e)$ corresponds to a $b$-bit 0-1 string $[w(e)]$ which is its binary representation, with the rightmost bit being the lowest bit.

To distinguish between values and strings, string $s$ is writen as $[s]$.
$LCP([w_1], [w_2])$ (which is also a binary string) is the longest common prefix of $[w_1]$ and $[w_2]$. For example $LCP([100],[101])=[10]$, and $LCP([000],[100])=[\ ]$(empty string).
$|[w]|$ denotes the length of the string $[w]$, and $[x][y]$ denotes the concatenation of two strings $[x]$ and $[y]$.
$[x] < [y]$ means that the lexicographical order of $[x]$ is smaller. For example, $[0111]<[101]<[1010]$. 

We define $E_{[x]} = \{e \in E \ \mid \  [w(e)]\text{ has prefix }[x]\}$ to be the set of all edges whose weight has prefix $[x]$ (also call the edges have prefix $[x]$), and similarly the subgraph $G_{[x]} = (V, E_{[x]})$. For convenience, an edge set or a subgraph with subscript $[x]$ means all of the edges in it have prefix $[x]$.

\paragraph*{Rectangular Matrix Multiplication}

We use $M(m,n,p)$ to denote the asymptotic time complexity of multiplying an $m \times n$ matrix with an  $n \times p$ matrix. We denote $M(n,n,n)=O(n^{\omega})$, where $\omega < 2.373$~\cite{Coppersmith:1987:MMV:28395.28396,williams2012multiplying,le2014powers}. In this paper, rectangular matrix multiplications are straightforwardly reduced to square matrix multiplications. So we will only use the following fact. (Here $\min(x_1,\cdots, x_k)$ means the minimum one in $\{x_1,\cdots, x_k\}$.)

\begin{lemma} \label{L0}
$M(m,n,p) = O\left(\dfrac{mnp}{\min(m,n,p)^{3 - \omega}}\right)$.
\end{lemma}
\begin{proof}
The rectangular matrices are decomposed into $\min(m,n,p) \times \min(m,n,p)$ square matrices, then standard fast square matrix multiplication is applied.
\end{proof}
\subsection{Na\"ive Algorithm}

Let us first take a look at the na\"ive algorithm for APNP, a simple Dijkstra-type search~\cite{dijkstra1959note}. Initially set $d(i,j)=w(i,j)$ for all edges $(i,j)\in E$, and $d(i,j)=+\infty$ otherwise. Each time we visit the minimum unvisited $d(i,j)$ and enumerate every out-going edge $(j,k)$ of vertex $j$. If $d(i,j)$ and $(j,k)$ form a nondecreasing path, then we update $d(i,k) \gets \min(d(i,k), w(j,k))$. (See Algorithm \ref{algonaive}.)  We refer to this update step as \textit{relaxing} edge $(j,k)$ w.r.t. $d(i,j)$. By the greedy nature of Dijkstra search, when we visit $d(i,j)$, $d(i,j) = OPT(i,j)$. Namely, it is optimal. 

\begin{algorithm}
\caption{Relaxing edges na\"ively}\label{algonaive}
\begin{algorithmic}[1]
\While{there exists unvisited $d(i,j)$}
\State{visit the minimum unvisited $d(i,j)$}
\For{every $(j,k)$ such that $w(j,k) > d(i,j)$}
\State{perform update $d(i,k)\gets \min(d(i,k), w(j,k))$ }
\EndFor
\EndWhile
\end{algorithmic}
\end{algorithm}

For clarity, our usages of symbols $i$, $j$, $k$ stick to the following convention: $d(i,k)$ refers to the path being updated by the concatenation of path $d(i,j)$ and edge $(j,k)$.

The na\"ive algorithm takes $O(n^3)$ time when edge weights are integers from 0 to $2^b-1$, since a bucket heap of size $O(n^2)$ is enough to maintain unvisited $d(i,j)$. 
\subsection{Classifying edges according to degrees}\label{subsecclassifying}

In order to avoid relaxing all edges when visiting $d(i,j)$, our algorithm will classify the edges based on the degrees of their two endpoints, and partition the edges whose both endpoints have high degrees into two sets by the order of their weights, then recursively deal with these two sets. We relax different types of edges by different approaches. First, we define ``high-degree'' and ``low-degree'' vertices in a subgraph. 


\paragraph*{High degree and low degree}

In this paper, sometimes we use a subset of edges $E'\subseteq E$ to denote the subgraph $G'=(V',E')$  where $V'$ is the set of vertices associated with $E'$, namely, vertices in $E'$ refers to vertices in $V'$. In a subgraph $G'=(V',E')$ of $G$, a vertex has high outdegree if its outdegree is larger than $n^{1 - t}$, and otherwise, it has low outdegree. Here $t$ is a parameter to be determined later (we will choose $t = \frac{3 - \omega}{2}$). Similarly, a vertex has high indegree if its indegree is larger than $n^{1 - t}$, and otherwise, it has low indegree. In our algorithm, edges $(j,k)$ in a subgraph are divided into three types based on the outdegree of $j$ and indegree of $k$:
\begin{itemize}
	\item Low edge: if $j$ has low outdegree
	\item High-high edge: if $j$ has high outdegree and $k$ has high indegree 
	\item High-low edge: if $j$ has high outdegree and $k$ has low indegree
\end{itemize}

\paragraph*{Divide the edges}
In this binary partition procedure, starting from the entire edge set $E'_{[ \ ]}=E$, each time we divide the set of high-high edges $H_{[x]}$ in $E'_{[x]}$ into two parts: $H'_{[x][0]}$ and $H'_{[x][1]}$, based on its next bit after prefix $[x]$, then recursively partition the edge sets $H'_{[x][0]}$ and $H'_{[x][1]}$. Here we use $L_{[x]}$ and $\Gamma_{[x]}$ to denote the low edges and high-low edges with prefix $[x]$ obtained in the algorithm. 

\begin{algorithm}
\caption{Divide the edges}\label{algodivide}
\begin{algorithmic}[1]
\Procedure{Divide}{$E'_{[x]}$}
    \State Consider the indegrees and outdegrees of vertices in the graph $G'_{[x]}=(V,E'_{[x]})$;
    \State Let $L_{[x]}$ be the set of edges from low outdegree vertices;
    \State Let $H_{[x]}$ be the set of edges from high outdegree vertices to high indegree vertices;
    \State Let $\Gamma_{[x]}$ be the set of edges from high outdegree vertices to low indegree vertices;
    \State Let $H'_{[x][0]}=\{(j,k)\in H_{[x]}\text{ s.t. $[w(j,k)]$ has prefix $[x][0]$}\}$ and $H'_{[x][1]}=\{(j,k)\in H_{[x]}\text{ s.t. $[w(j,k)]$ has prefix $[x][1]$}\}$.
    \State \Call{Divide}{$H'_{[x][0]}$}
    \State \Call{Divide}{$H'_{[x][1]}$}
\EndProcedure
\end{algorithmic}
\end{algorithm}

At the beginning we call $\textsc{Divide}(E)$.
Since only the edges in $H_{[x]}$ go into the next recursion, every edge can only appear in one of $\{L_{[x]}\}$ or $\{\Gamma_{[x]}\}$ but can appear in many of $\{H_{[x]}\}$.

\subsection{Outline of our algorithm}

In our algorithm, we run a Dijkstra-type procedure. When visiting $d(i,j)$, which is guaranteed to be optimal, we relax all the edges $(j,k)$ in $L_{[x]}$ and $\Gamma_{[x]}$ w.r.t. $d(i,j)$ for all $[x]$ which is a prefix of $[d(i,j)]$. The outdegree of $j$ is small in $L_{[x]}$ but not in $\Gamma_{[x]}$, so to relax $L_{[x]}$, we relax all edges from $j$ as the na\"ive algorithm. But to relax $\Gamma_{[x]}$, a preprocessing procedure of the edges in $\Gamma_{[x]}$ is needed in order to save time. For edges in $H_{[x]}$, instead of immediately relaxing them, we wait until all optimal paths with prefix $[x][0]$ are visited, then perform a $(\min, \le)$-matrix product of these paths with the adjacency matrix of $H'_{[x][1]}$ to relax all edges in it, so all high-high edges are relaxed w.r.t. the paths whose the longest common prefix with the edges is $[x]$. 
Details of these methods will be given in Section~\ref{BT}.

\section{Basic techniques} \label{BT}

In this section, we explain the method we use for relaxing edges. Suppose we are trying to relax edges $(j,k)$ w.r.t. $d(i,j)$, it will be based on whether $(j,k)$ is in $L_{[x]}$, $H_{[x]}$, or $\Gamma_{[x]}$.

As explained before, $L_{[x]}$ is easiest. Since $j$ has low outdegree, like the na\"ive algorithm, we simply relax all of its outgoing edges, which takes $O(n^{1-t})$ time. $H_{[x]}$ and $\Gamma_{[x]}$ are handled by different methods, though they both use the ``row/column balancing'' technique of matrix proposed in \cite{duan2009fast}. Instead of considering it as splitting rows/columns on matrix, we describe the balancing technique as splitting vertices so that every vertex has low outdegree/indegree. In the following subsections we describe the balancing technique and how to handle $H_{[x]}$ and $\Gamma_{[x]}$.

\subsection{Balancing}\label{section:balancing}

A graph $G = (V,E)$ contains at most $\frac{|E|}{n^{1 - t}}$ vertices with high indegrees or outdegrees.
If there are exactly $\Theta(\frac{|E|}{n^{1 - t}})$ number of vertices with high degrees, we would expect an average degree of $O(n^{1 - t})$.
However, there may be less vertices with high degrees, and the degrees of some vertices may be far greater than $n^{1-t}$.
To balance the indegree (outdegree) of each vertex, we split every vertex into several vertices each of indegree (outdegree) $n^{1- t}$ and one vertex of indegree (outdegree) $\leq n^{1-t}$ . The number of new vertices with indegree (outdegree) exactly $n^{1 - t}$ is bounded by $O\left(\frac{|E|}{n^{1 - t}}\right)$, and every original high indegree (outdegree) vertex corresponds to at most one new vertex with indegree (outdegree) $<n^{1 - t}$, thus we have at most $O\left(\frac{|E|}{n^{1 - t}}\right)$ many new vertices each with indegree (outdegree) $\leq n^{1 - t}$.


In our algorithm, for edge set $\{(j,k)\}$, we use this technique to either balance the outdegrees of vertices $j$ or the indegrees of vertices $k$. Here we demonstrate this technique for balancing indegrees of $k$ as an example in Table~\ref{balancing}. Denote the set of edges after balancing to be $\bar{E}$, then the graph corresponding to $\bar{E}$ is actually a bipartite graph with edges between vertices $j$ and $k'_r$. The procedure for balancing outdegrees of $j$ is symmetric. 

\begin{table}
\centering
\begin{tabular}{|p{12.8cm}|}
\hline
\textbf{Balancing the indegrees of vertices}\\
\hline 
(S1) For every $k$, sort all the in-coming edges of $k$ in increasing order and obtain the sorted list $L_k$.\\
(S2) Split $k$ into vertices $k'_1, k'_2, \cdots k'_p$, and divide $L_k$ into segments each of size $n^{1 - t}$ (while the last segment is possibly incomplete). The edges in the $r$-th segment is assigned to vertex $k'_r$. Namely, edge $(j,k)$ in the $r$-th segment now becomes edge $(j, k'_r)$. \\
(S3) Let $[L(k'_r), R(k'_r)]$ denote the weight range of in-coming edges of $k'_r$. Namely $L(k'_r)$ is the minimum weight of in-coming edges of $k'_r$, and $R(k'_r)$ is the maximum. Obviously, $L(k'_1) \leq R(k'_1) < L(k'_2) \leq R(k'_2) < \cdots < L(k'_p) \leq R(k'_p)$.\\
\hline
\end{tabular}
\caption{Balancing the indegrees of vertices $k$}
\label{balancing}
\end{table}

\subsection{High-high edges} \label{high-high}

The technique for high-high edges solves the following problem: Given a set $P$ of optimal paths of prefix $[x][0]$ and a set $H'_{[x][1]}$ of high-high edges, the problem asks to relax all edges in $H'_{[x][1]}$ w.r.t. paths in $P$, namely, extend paths in $P$ by a single edge in $H'_{[x][1]}$. 

This problem is equivalent with a length-two nondecreasing path problem. 
As discussed in Section~\ref{intro}, this can be solved by fast $(\min, \le)$-product implied in~\cite{duan2009fast}. But the rectangular version is not covered by~\cite{duan2009fast}, so we fully describe the algorithm and its analysis (in graphs).

We have two extra guarantees when we use this procedure in our main algorithm:
\begin{itemize}
	\item $P$ is the set of optimal paths $OPT(i,j)$ such that $[OPT(i,j)]$ has prefix $[x][0]$, so any path $d(i,j)$ in $P$ can form a nondecreasing path with any edge $(j,k)$ in $H'_{[x][1]}$.
	\item All $d(i,j)$ satisfying $[d(i,j)] < [x][1]$ are already visited by Dijkstra search. So we can tell whether $[OPT(i,j)] < [x][1]$ or not.
\end{itemize}

Suppose there are $n_{[x][1]}$ many $[OPT(i,j)]$ which have $[x][1]$ as a prefix. In our algorithm the time complexity will depend on $n_{[x][1]}$. 

The first step is to apply the balancing technique in Section~\ref{section:balancing} to indegrees of vertices $k$ in $H'_{[x][1]}$, and let the edge set after balancing be $\bar{H'}_{[x][1]}$. 
Here each high indegree vertex $k$ is split into vertices $\{k'_r\}$. We define the following two matrices: ($j$ and $k'_r$ only include vertices having out-going edges and in-coming edges in $\bar{H'}_{[x][1]}$, respectively.)
\begin{align*}
&A_{i,j} = \begin{cases}
1& \text{$d(i,j) \in P$}\\
0& \text{otherwise}
\end{cases} &B_{j,k'_r} = \begin{cases}
1& \text{$(j,k'_r) \in \bar{H'}_{[x][1]}$}\\
0& \text{otherwise}
\end{cases}
\end{align*}

Then we multiply them with rectangular matrix multiplication.
Let $C = AB$, then, 
\begin{align*}
&C_{i,k'_r} \begin{cases}
>0 & \text{if there is a nondecreasing path from $i$ to $k'_r$} \\
=0 & \text{otherwise}
\end{cases}
\end{align*}

For every pair $(i,k)$ such that $[OPT(i,k)] \geq [x][1]$ ($d(i,j)\not\in P$) and $k$ is an in-coming vertex in $H'_[x][1]$, we find the minimum $r$ with $C_{i,k'_r} > 0$ and relax all $n^{1 - t}$ incoming edges $(j,k'_r)$ of $k'_r$ w.r.t. $d(i,j)$ if $d(i,j)\in P$. We can skip other $r' > r$ because $OPT(i,k) \leq R(k'_r) < L(k'_{r'})$. Namely, we only have to relax $O(n^{1 - t})$ many edges for each $(i,k)$, which is the benefit of balancing. We attribute this $O(n^{1 - t})$ cost of relaxations to the $OPT(i,k)$ with prefix $[x][1]$ found in these relaxations, which must exist if $C_{i,k'_r}>0$. 

The procedure above consists of three parts: finding minimum $r$ for each $(i,k)$, relaxation and matrix multiplication.
The time cost for the first part is at most enumerating all nonzero elements of $C$, so it is dominated by the matrix multiplication part.

\begin{lemma}\label{high-high:enumeration}
The relaxation part takes $O\left(n_{[x][1]} n^{1 - t}\right)$ time in total.
\end{lemma}
\begin{proof}
The cost of relaxation is attributed to each $OPT(i,k)$ with prefix $[x][1]$. Since there are $n_{[x][1]}$ many $OPT(i,k)$ with prefix $[x][1]$, and each corresponds to $O(n^{1-t})$ relaxations, it costs $O(n_{[x][1]} n^{1 - t})$ time in total.
\end{proof}

\begin{lemma}\label{high-high:multiplication}
In the matrix multiplication part, $A$ is an $n\times O\left(\min\left(n, \frac{2^{b - |[x]|}}{n^{1 - t}}\right)\right)$ matrix, and $B$ is a $O\left(\min\left(n, \frac{2^{b - |[x]|}}{n^{1 - t}}\right)\right)\times O\left(\frac{2^{b - |[x]|}}{n^{1 - t}}\right)$ matrix, so the time complexity for matrix multiplication is $M\left(n,\min\left(n,\frac{2^{b - |[x]|}}{n^{1 - t}}\right),\frac{2^{b - |[x]|}}{n^{1 - t}}\right)$.
\end{lemma}
\begin{proof}
There are at most $\min\left(n,\frac{|H'_{[x][1]}|}{n^{1 - t}}\right)$ many $j$ because vertices $j$ have high outdegrees. After balancing, there are at most $O\left(\frac{|H'_{[x][1]}|}{n^{1 - t}}\right)$ many $k'_r$ by discussion in Section \ref{section:balancing}. Since each edge in $H'_{[x][1]}$ has prefix $[x][1]$, $|H'_{[x][1]}| \leq 2^{b - |[x]| - 1} = O\left(2^{b - |[x]|}\right)$. Plug in the size of $H'_{[x][1]}$ gives the desired bound. 
\end{proof}

\subsection{High-low edges} \label{high-low}

Now we consider the relaxation of the high-low edges in $\Gamma_{[x]}$ when visiting paths with the same prefix $[x]$. To preprocess high-low edges, we run an initialization step when all optimal paths less than $[x]$ have been visited. As before, we denote the number of $[OPT(i,j)]$ with prefix $[x]$ by $n_{[x]}$.

Since the outdegrees of vertices $j$ are high, we cannot relax edges one by one. But we still want to utilize the property that the indegrees of $k$ are low. As in the last subsection, we denote the set of optimal paths we have found to be $P$, namely when we visit $d(i,j)$, we add $d(i,j)$ to $P$. By the nature of Dijkstra search, such $d(i,j)$ is always visited in increasing order. At the initialization step, $P$ is the set of all optimal paths less than $[x]$. During the procedure, optimal paths with prefix $[x]$ are added to $P$.

We also maintain a dynamic set $Q$ which is initially empty. When we relax an edge $(j,k)$ w.r.t. $d(i,j)$, we put $d(i,k)$ into $Q$ only if $[d(i,k)]\geq [x]$, that is, $d(i,k)$ was not visited at initialization. So $Q$ contains new nondecreasing path but not guaranteed to be optimal.

\paragraph*{Initialization by Matrix Multiplication}

The first step is to apply the balancing technique to the outdegrees of vertices $j$ in $\Gamma_{[x]}$ such that every vertex $j$ in $\Gamma_{[x]}$ is split into a sequence of vertices $\{j'_r\}$. Suppose the edge set after balancing is $\bar{\Gamma}_{[x]}$. Then, we define the following two matrices: ($j'_r$ and $k$ only include vertices having out-going edges and in-coming edges in $\bar{\Gamma}_{[x]}$, respectively.)
\begin{align*}
& A_{i,k} = \begin{cases}
1& \text{$d(i,k) \not\in P \cup Q$ }\\
0& \text{otherwise}
\end{cases} 
&B_{k,j'_r} = \begin{cases}
1& \text{$(j'_r, k) \in \bar{\Gamma}_{[x]}$ }\\
0& \text{otherwise}
\end{cases}
\end{align*}

We compute $C = AB$. Basically speaking, the matrix $C$ can indicate whether we need to relax edges $(j'_r,k)$ from $j'_r$ when visiting $d(i,j)$. Note that when we run initialization, $Q$ is empty, so $d(i,k)\not\in Q$ is trivially true. But we will add paths to $Q$ and dynamically update the matrix multiplication later.

We make the following observation about $C_{i,j'_r}$. 

\begin{observation}\label{Cij}
If $d(i,j) < L(j'_r)$ and $C_{i,j'_r} > 0$, there is at least one edge $(j'_r, k)$ such that $d(i,k) \not\in P \cup Q$ and $[x]$ is a prefix of $[OPT(i,k)]$. \\
Conversely, if $C_{i,j'_r} = 0$, there is no such an edge $(j'_r, k)$.
\end{observation}
\begin{proof}
Since $C_{i, j'_r} > 0$, at least one vertex $k$ satisfy the following:
\begin{itemize}
	\item $A_{i,k} > 0$ : $d(i,k) \not\in P$, so $[OPT(i,k)] \geq [x][0\cdots0]$. Also $d(i,k) \not\in Q$. 
	\item $B_{k, j'_r} > 0$ : $(j'_r, k) \in \bar{\Gamma}_{[x]}$. 
\end{itemize}
Since $d(i,j) < L(j'_r)$,  $d(i,j)$ and the original edge of $(j'_r,k)$ form a non-decreasing path. So $[OPT(i,k)] \leq [x][1\cdots1]$. 

Conversely, if $C_{i,j'_r} = 0$, for every $k$, at least one of these happens:
\begin{itemize}
	\item $A_{i,k} = 0$ : $d(i,k) \in P \cup Q$
	\item $B_{k,j'_r} = 0$ : edge $(j'_r, k)$ does not exist in $\bar{\Gamma}_{[x]}$. 
\end{itemize}
\end{proof}

\paragraph*{Update matrix multiplication}

When a new path $d(i,k)$ is added to $P$ or $Q$, the matrix $A$ and product $C$ need to be updated. Adding a new path to $P$ or $Q$ only changes one entry of $A_{i,k}$, so we utilize the low indegree of $k$. There are at most $O(n^{1 - t})$ many $j'_r$ such that $B_{k, j'_r} \not= 0$, since $\Gamma_{[x]}$ contains high-low edges only. So the update of $C$ when changing one element $A_{i,k}$ takes only $O(n^{1-t})$ time by enumerating all nonzero $B_{k,j'_r}$. This cost can be attributed to each $[OPT(i,k)]$ with prefix $[x]$, since every such path can only be added to $P\cup Q$ once. 
However, adding $d(i,k)$ to $Q$ does not mean we have found the optimal path $OPT(i,k)$, as it can still be updated. How to deal with this will be discussed later. 

\paragraph*{Relaxation when visiting $d(i,j)$}

When we visit $d(i,j)$, for each split vertex $j'_r$ of $j$, there are three cases:

\begin{enumerate}
\item $d(i,j) > R(j'_r)$ : $d(i,j)$ cannot form a non-decreasing path with any out-going edge of $j'_r$, so we skip $j'_r$. 
\item $d(i,j) \in [L(j'_r), R(j'_r)]$ : We relax all out-going edges of $j'_r$ larger than $d(i,j)$. 
\item $d(i,j) < L(j'_r)$ : Only when $C_{i,j'_r} > 0$, we relax all out-going edges of $j'_r$ one by one. 
\end{enumerate}

By Observation \ref{Cij}, when $C_{i,j'_r} = 0$, for each edge $(j'_r,k)$, either $d(i,k) \in P$ or $d(i,k) \in Q$. If $d(i,k) \in P$, it needs no more update. If $d(i,k) \in Q$, roughly speaking, since $k$ has indegree less than $n^{1 - t}$, the updates for $d(i,k)$ can be done ``in advance'' when it is added to $Q$. The details will be clear later. This is why we can skip $j'_r$ when $C_{i,j'_r} = 0$.

Since the degree of $j'_r$ is bounded by $n^{1 - t}$, the relaxation takes $O(n^{1 - t})$ time for each $j'_r$. For the second case, it only happens once for each $d(i,j)$ because $[L(j'_r), R(j'_r)]$ are disjoint for different $j'_r$, so the $O(n^{1 - t})$ cost is attributed to $d(i,j)$. For the third case, by Observation~\ref{Cij}, there is at least one edge $(j'_r,k)$ such that $OPT(i,k) \not\in P\cup Q$ with prefix $[x]$, so the $O(n^{1 - t})$ time is attributed to $d(i,k)$. (If there is more than one such $k$, choose an arbitrary one.) Then $d(i,k)$ is added to $Q$, and the cost of updating $A$ and $C$ is also bounded by $n^{1 - t}$, dominated by the cost of relaxation.

Because the first path $d(i,k)$ we found for each $(i,k)$ is not necessarily the optimal one, we discuss how to handle all future updates of $d(i,k)$ ``in advance'' after adding it to $Q$. We enumerate every in-coming edge $(j'',k)\in \Gamma_{[x]}$ of $k$. If $d(i,j'')$ is not in $P$, we add $(i,k)$ to a \textit{waiting list} for $(i,j'')$, denoted by $W(i,j'')$. When $d(i,j)$ is visited in the future, we can go through its waiting list $W(i,j)$ and update $d(i,k)$ for all pair $(i,k)$ in the list. There are only $n^{1 - t}$ in-coming edges for $k$, so the waiting list construction cost is also $O(n^{1 - t})$ for every $d(i,k)$.


In conclusion, we follow the procedure in Algorithm \ref{algoupd} when visiting $d(i,j)$ with prefix $[x]$.

\begin{algorithm}
\caption{High-low relaxation when visiting $d(i,j)$}\label{algoupd}
\begin{algorithmic}[1]
\State{Add $d(i,j)$ to $P$ and update $A$ and $C=AB$.}
\For {$(i,k)$ in the \textit{waiting list} $W(i,j)$}
    \State{Relax $(j,k)$ w.r.t. $d(i,j)$ if $w(j,k)>d(i,j)$} \label{WL-relaxation}
\EndFor
\For{every $j'_r$ satisfying $d(i,j) \in [L(j'_r), R(k'_r)]$ or ($d(i,j) < L(j'_r)$ and $C_{i,j'_r} > 0$)}
        \For {every outgoing edge $(j'_r,k)$ of $j'_r$ larger than $d(i,j)$}
            \If{$d(i,k)\not\in P\cup Q$}
                \State{Relax (the original edge of) $(j'_r,k)$ w.r.t. $d(i,j)$}
                \State{Add $d(i,k)$ to $Q$ and update $A$ and $C=AB$}
                \For {incoming edge $(j'',k)$ of $k$}
                    \State{Add $(i,k)$ to the \textit{waiting list} $W(i,j'')$ if $d(i,j'')\not\in P$} \label{addtowl}
                \EndFor
            \EndIf
        \EndFor
\EndFor
\end{algorithmic}
\end{algorithm}


\paragraph*{Complexity}

This procedure is divided into matrix multiplication part (initialization) and relaxation part (Algorithm \ref{algoupd}) as well.

\begin{lemma} \label{high-low:enumeration}
The relaxation part takes $O\left(n_{[x]}n^{1- t}\right)$ time.
\end{lemma}
\begin{proof}
From discussion above, the $O(n^{1 - t})$ cost of each relaxation is either attributed to optimal $d(i,j)$ with prefix $[x]$ or $d(i,k) \in Q$. For each $d(i,k) \in Q$, $OPT(i,k)$ is of course larger than $[x][0\cdots0]$, and  then relaxed by an edge $\leq [x][1\cdots1]$, so the size of $Q$ is also bounded by $n_{[x]}$. Since each $d(i,j)$ can only be added to $P$ and $Q$ once, respectively, the total time is  $O\left(n_{[x]}n^{1- t}\right)$.


The total size of all waiting lists $W(i,j)$ is bounded by $O\left(n_{[x]}n^{1- t}\right)$ as well, because each time when an $d(i,k)$ is added to $Q$, we enumerate $\leq n^{1 - t}$ many incoming edges of $k$, and add $(i,k)$ to waiting list at Line~\ref{addtowl}. Every waiting list can only be relaxed once, thus, the relaxation of waiting list edges in Line~\ref{WL-relaxation} needs $O\left(n_{[x]}n^{1- t}\right)$ in total. 
\end{proof}

The following lemma is crucial. Although edges $(j,k)$ in $\Gamma_{[x]}$ are high-low edges, the number of $k$ is not directly bounded, but remind that in our binary partition of edges, only high-high edges of previous level can be in this set, thus in fact the number of $k$ \emph{cannot} be asymptotically larger than the number of $j'_r$.
\begin{lemma} \label{high-low:multiplication}
In the matrix multiplication part, $A$ is an $n\times O\left(\min\left(n, \frac{2^{b - |[x]|}}{n^{1 - t}}\right)\right)$ matrix, and $B$ is an $O\left(\min\left(n, \frac{2^{b - |[x]|}}{n^{1 - t}}\right)\right) \times O\left(\frac{2^{b - |[x]|}}{n^{1 - t}}\right)$ matrix, so the time complexity of matrix multiplication is $M\left (n, \min\left(n, \frac{2^{b - |[x]|}}{n^{1 - t}}\right), \frac{2^{b - |[x]|}}{n^{1 - t}}\right)$.
\end{lemma}
\begin{proof}
Since $|\Gamma_{[x]}| \leq 2^{b - |[x]|}$ (because each edge in it has prefix $[x]$), after balancing, there are at most $O\left(\frac{2^{b - |[x]|}}{n^{1 - t}}\right)$ many $j'_r$. 

If $[x]\neq []$, suppose $[x] = [x'][0/1]$, namely $[x']$ is the prefix of $[x]$ which is one bit shorter. If $(j,k) \in \Gamma_{[x]}$, by Algorithm~\ref{algodivide}, $(j,k) \in H_{[x']}$. Since $k$ has high indegree in $H_{[x']}$, the number of such $k$ is bounded by $O\left(\min\left(n,\frac{|H_{[x']}|}{n^{1 - t}}\right)\right)$. Also $|H_{[x']}| = O\left(2^{b - |[x]|}\right)$. Plug it in gives the $O\left(\min\left(n, \frac{2^{b - |[x]|}}{n^{1 - t}}\right)\right)$ bound for the number of $k$. If $[x]=[]$, of course the number of $k$ is bounded by $n$.
\end{proof}

\section{Main algorithm for directed graphs and analysis} \label{AA}

\subsection{Main algorithm}

Just like in the na\"ive algorithm, we use a bucket to maintain all $d(i,j)$ we have found, and the minimal \emph{unvisited} $d(i,j)$ is guaranteed to be optimal.
Our algorithm enumerates the value $x$ from $0$ to $2^b - 1$ and visit $d(i,j)$ if $d(i,j) = x$. We carefully combine techniques introduced in the previous section into this framework. 

Recall that in Algorithm~\ref{algodivide} of Section~\ref{subsecclassifying} we define $L_{[x]}$ to be the set of low edges (from low outdegree vertices) and $\Gamma_{[x]}$ to be the set of high-low edges in $H'_{[x]}$, which are high-high edges in higher level, then divide the edge set $H_{[x]}$ of high-high edges in $H'_{[x]}$ to $H'_{[x][0]}$ and $H'_{[x][1]}$ and recursively deal with them. Our main algorithm is presented in Algorithm \ref{algomain}.

\begin{algorithm}
\caption{Main algorithm}\label{algomain}
\begin{algorithmic}[1]
\State{$d(i,j)=w(i,j)$ for all edges $(i,j)\in E$, and $d(i,j)=+\infty$ otherwise} \label{alg:initialize}
\For {$x$ from $0$ to $2^b - 1$}
	\For{all prefix $[y]$ of $[x]$}
		\If{$[x] = [y][000\cdots0]$}
		    \State{Run high-low edge initialization for edges in $\Gamma_{[y]}$} \label{alg:high-low:init}
		\EndIf
		\If {$[x] = [y][100\cdots0]$}
			\State{$P_{[y][0]} = \{d(i,j) \mid [y][0]\text{ is a prefix of }[d(i,j)]\}$}
			\State{Append high-high edges in $H'_{[y][1]}$ to paths in $P_{[y][0]}$} \label{alg:high-high}
		\EndIf
	\EndFor
	\For{$d(i,j) = x$}
	    \State{Mark $d(i,j)$ as visited, add $d(i,j)$ to $P$.}
	    \For {all prefix $[y]$ of $[d(i,j)]$}
	        \State{Relax edges $(j,k)\in L_{[y]}$} \label{alg:low}
	        \State{Relax edges $(j,k)\in \Gamma_{[y]}$ by Algorithm~\ref{algoupd}}\label{alg:high-low:upd}
	    \EndFor
	\EndFor
\EndFor
\end{algorithmic}
\end{algorithm}

When visiting an optimal path $d(i,j)$, we need to relax all edges $(j,k)$ which are larger than $d(i,j)$. 
\begin{observation} \label{threecases}
For $[d(i,j)] = [x]$, every edge $(j,k)$ larger than $d(i,j)$ must be one of the following three cases: (so $[y]$ is a prefix of both $[x]$ and $[w(j,k)]$.)
\begin{itemize}
    \item $(j,k)\in L_{[y]}$ for some prefix $[y]$ of $[x]$
    \item $(j,k)\in \Gamma_{[y]}$ for some prefix $[y]$ of $[x]$
    \item $(j,k)\in H'_{[y][1]}$ where $[y][0]$ is a prefix of $[x]$
\end{itemize}
\end{observation}
\begin{proof}
Consider the longest common prefix $[y]=LCP([x],[w(j,k)])$. If $(j,k)$ is not in the $L_{[y']}$ or $\Gamma_{[y']}$ for any prefix $[y']$ of $[y]$, it must be in $H_{[y]}$. Since $[y]$ is the longest common prefix and $w(j,k)$ is larger than $d(i,j)$, $[d(i,j)]$ has prefix $[y][0]$ and $[w(j,k)]$ has prefix $[y][1]$, thus $(j,k)$ is in $H'_{[y][1]}$.
\end{proof}

Thus, we can simply relax the edges of the first type, and use the method in Section~\ref{high-high} to relax the edges in $H'_{[y][1]}$ when all of the optimal paths with prefix $[y][0]$ have been visited.
The method for high-low edges $\Gamma_{[y]}$ is like a dynamic data structure: we initialize it when $[x]=[y][0\cdots 0]$, and update it when relaxing an edge in $\Gamma_{[y]}$.


\paragraph*{High-high edges}

For high-high edges, when $[x] = [y][100\cdots0]$, before those $d(i,j) = x$ are visited, we append edges in $H'_{[y][1]}$ to paths in $P_{[y][0]} = \{d(i,j) \mid [y][0]\text{ is a prefix of }[d(i,j)]\}$ using the technique introduced in Section \ref{high-high}. See Line \ref{alg:high-high}, Algorithm \ref{algomain}. 

In Section~\ref{high-high}, we have two guarantees. Now we check them one by one:

\begin{itemize}
	\item Because each edge in $H'_{[y][1]}$ has prefix $[y][1]$, and each path in $P_{[y][0]}$ has prefix $[y][0]$, the maximum weight in $P_{[y][0]}$ is smaller than the minimum weight of $H'_{[y][1]}$. At the time of $[x] = [y][100\cdots0]$, all paths in $P_{[y][0]}$ are optimal.
	\item Since $[x] = [y][100\cdots0]$, all $[d(i,j)] < [y][1]$ are visited, and none of $[d(i,j)] \geq [y][1]$ are visited yet. 
\end{itemize}

\paragraph*{High-low edges}




We initialize for $\Gamma_{[y]}$ when $[x] = [y][000\cdots0]$ before we visit those $d(i,j) = x$. See Line \ref{alg:high-low:init} of Algorithm \ref{algomain}. All $[d(i,j)] < [y]$ are visited, and none of $[d(i,j)] \geq [y]$ are visited. Once a $d(i,j)$ within the range $[y][000\cdots0] \sim [y][111\cdots1]$ is visited, we use the approach in Algorithm~\ref{algoupd}.

\subsection{Correctness}

We now prove the correctness of our algorithm. Suppose the last edge of $OPT(i,k)$ is $(j,k)$. Then $d(i,k)$ is correctly computed before it is visited if and only if the following two conditions holds: 
\begin{itemize}
    \item If $i\neq j$, $d(i,j)$ is correctly computed before it is visited. 
    \item After $d(i,j)$ is visited, before we visit $d(i,k)$, $d(i,k)$ is updated by relaxing the edge $(j,k)$ w.r.t. $d(i,j)$.  
\end{itemize}

We prove the second condition holds for every $d(i,j)$, and the first one simply follows from induction. 

Suppose $[z] = LCP([OPT(i,j)], [OPT(i,k)]) = LCP([OPT(i,j)], [w(j,k)])$. By Observation~\ref{threecases} the last edge $(j,k)$ must be in one of the three cases, so we check them one by one. 

\begin{lemma} \label{Correct1}
If $(j,k)$ is in $L_{[y]}$ for some prefix $[y]$ of $[z]$, and $d(i,j)$ is correctly computed before visited, then $d(i,k)$ is also correctly computed before visited. 
\end{lemma}
\begin{proof}
Because $[y]$ is also a prefix of $[OPT(i,j)]$, at Line \ref{alg:low}, when we visit $d(i,j)$, $d(i,k)$ is updated by relaxing $(j,k)$. Because $OPT(i,j) < OPT(i,k)$, $d(i,k)$ is not visited yet. 
\end{proof}

\begin{lemma} \label{Correct2}
Suppose $(j,k)$ is in $H'_{[y][1]}$ where $[y][0]$ is a prefix of $[OPT(i,j)]$. If $d(i,j)$ is correctly computed before visited, $d(i,k)$ is also correctly computed before visited.
\end{lemma}
\begin{proof}
We can see $[y] = [z]$ from the proof of Observation \ref{threecases}. At Line~\ref{alg:high-high}, when $[x] = [y][100\cdots0]$, $d(i,k)$ is updated by $d(i,j)$ and $(j,k)$. $d(i,j)$ is already visited before because it has prefix $[y][0]$. $d(i,k)$ will be visited later because it has prefix $[y][1]$.
\end{proof}

\begin{lemma} \label{Correct3}
If $(j,k)$ is in $\Gamma_{[y]}$ for some prefix $[y]$ of $[z]$, and $d(i,j)$ is correctly computed before visited, then $d(i,k)$ is also correctly computed before visited.
\end{lemma}
\begin{proof}
Since $[y]$ is a prefix of both $[OPT(i,j)]$ and $[w(j,k)]$, the initialization for $\Gamma_{[y]}$ is done at Line~\ref{alg:high-low:init} when $[x] = [y][000\cdots0]$. 
After that, $d(i,k)$ is updated when we visit $d(i,j)$ at Line \ref{alg:high-low:upd}. $d(i,k)$ is not visited yet because $OPT(i,k) > OPT(i,j)$. 
\end{proof}

\begin{lemma}
All $d(i,j)$ are correctly computed before visited.
\end{lemma}
\begin{proof}
This follows from a simple induction. In the base case, for all length $1$ optimal paths $OPT(i,j)$, they are obviously correctly computed in Line~\ref{alg:initialize}. Then if all length $l - 1$ paths $OPT(i,j)$ are correctly computed before visited, by Lemma \ref{Correct1}, \ref{Correct2}, \ref{Correct3}, all length $l$ paths $OPT(i,j)$ are also correctly computed before visited. 
\end{proof}

\subsection{Running time}

\begin{lemma}
The relaxation for low edges ($L_{[x]}$) takes $\tilde{O}\left(n^{3 - t}\right)$ time in total. (Line \ref{alg:low}) 
\end{lemma}
\begin{proof}
At Line \ref{alg:low}, we only enumerate $O(n^{1 - t})$ many edges because $j$ has low outdegree in $L_{[y]}$. Since there are only $b = O(\log n)$ many prefix $[y]$ for each $[d(i,j)]$, each $d(i,j)$ takes $\tilde{O}(n^{1 - t})$ time. So in total, these updates take $\tilde{O}(n^{3 - t})$ time for all $O(n^2)$ many $d(i,j)$. 
\end{proof}

\begin{lemma}\label{high-high:enumeration-time}
The relaxation for high-high edges and high-low edges besides matrix multiplication takes $\tilde{O}\left(n^{3 - t}\right)$ time in total.
\end{lemma}
\begin{proof}
By Lemma \ref{high-high:enumeration} and Lemma \ref{high-low:enumeration}, for each $[y]$, the complexity for relaxation is bounded by $(n_{[y]} + n_{[y][1]})n^{1 - t} = O\left(n_{[y]} n^{1 - t}\right)$, where $n_{[y]}$ stands for the number of optimal paths with prefix $[y]$. Since an optimal path can be counted in $O(\log n)$ many $n_{[y]}$, the total time is therefore $\tilde{O}(n^{3 - t})$. 
\end{proof}

\begin{lemma}\label{high-high:matrix-multiplication-time}
The matrix multiplication parts for high-high edge updates and the initialization of high-low edge updates take $\tilde{O}\left(n^{t + \omega}\right)$ time in total. 
\end{lemma}
\begin{proof}
By Lemma~\ref{high-high:multiplication} and~\ref{high-low:multiplication}, the complexity for matrix multiplication is at most $M\left (n, \min\left(\frac{2^{b - |[y]|}}{n^{1 - t}}, n\right), \frac{2^{b - |[y]|}}{n^{1 - t}}\right)$ for each $[y]$. We fix the length of $[y]$, denoted by $l = |[y]|$, then consider the two cases:

\begin{itemize}
	\item $2^l < n^t$ : There are at most $2^l$ many such $[y]$, and each takes $$M\left(n, n, \frac{2^{b - l}}{n^{1 - t}}\right) = O\left(\frac{n^2 \cdot \frac{2^{b - l}}{n^{1 - t}}}{n^{3 - \omega}}\right) = O\left(n^{t + \omega} \cdot 2^{-l}\right)$$ This follows from both Lemma \ref{L0} and the fact that $2^b = |E| = O(n^2)$. For each $l$, the time complexity is exactly $O(n^{t + \omega})$. So the total complexity is $\tilde{O}(n^{t + \omega})$ since $l = O(\log_2(n))$. 
	\item $2^l \geq n^t$ : There are at most $2^l$ many such $[y]$. Each takes 
	\begin{align*}
	M\left(n, \frac{2^{b - l}}{n^{1 - t}}, \frac{2^{b - l}}{n^{1 - t}}\right) &= O\left(n \cdot \left(\frac{2^{b - l}}{n^{1 - t}}\right)^{2 - (3 - \omega)} \right) \\ &= O\left( n^{(t + 1)(\omega - 1) + 1} \cdot 2^{-l(\omega - 1)} \right) 
	\end{align*}
	So for all $[y]$ of length $l$, it takes $O\left( n^{(t + 1)(\omega - 1) + 1} \cdot 2^{-l(\omega - 2)} \right)$ time. The term $2^{-l(\omega - 2)}$ is maximized when $l$ is minimized, so $2^{-l(\omega - 2)} \leq n^{-t(\omega-2)}$, and the total time for all lengths of $[y]$ is
	\begin{equation*}
	\tilde{O}\left(n^{(t + 1)(\omega-1) + 1 - t(\omega - 2)} \right) = \tilde{O}\left(n^{t + \omega} \right)\end{equation*}
\end{itemize}
\end{proof}





\begin{theorem}
The \textit{All Pair Non-decreasing Paths} (APNP) problem on directed graphs can be solved in $\tilde{O}\left(n^{\frac{3 + \omega}{2}}\right)$ time. The optimal path of length $l$ between any two vertices can also be explicitly found in $O(l)$ time if we slightly modify the algorithm.
\end{theorem}
\begin{proof}
We choose $t = \frac{3 - \omega}{2}$. The running time of this algorithm follows from previous lemmas. Since all optimal paths $OPT(i,k)$ are obtained by relaxation of edges $(j,k)$, we can store the last edge $(j,k)$ for each $OPT(i,k)$, so retrieving the optimal path can be done in $O(l)$ time.
\end{proof}

\section{Near-optimal algorithm for undirected APNP}\label{undirectedsection}
In this section we show that the APNP problem on undirected graphs can be solved in time $\tilde O (n^2)$. Let $G = (V,E)$ be an undirected graph with $m=O(n^2)$ edges $e_1,e_2,\dots,e_m$. For simplicity we assume that the weights $w(e_i)$ of edges are distinct using Lemma \ref{equal} in Appendix~\ref{equal-weight}. But here we allow multiple edges in $G$.

\subsection{Basic algorithm}
Our algorithm proceeds as follows: Beginning from an empty graph, we successively insert the edges in ascending order of weights. An $n\times n$ zero-one matrix $A$ is maintained and updated after each insertion, where $A_{i,j}=1$ if an only if there exists a non-decreasing path $(i,j)$ consisting only of edges that are already inserted. If $A_{i,j}$ is updated from 0 to 1 when inserting edge $e$, we have $OPT(i,j) = w(e)$. Algorithm~\ref{alg:undirected1} below is a straightforward implementation of this algorithm in time $O(n^3)$, in which non-infinity elements of matrix $D$ denote the optimal answers.

\begin{algorithm}
\caption{Basic algorithm for undirected APNP in time $O(n^3)$} \label{alg:undirected1}
\begin{algorithmic}[1]
\State{Sort the edges so that $w(e_1)<w(e_2)<\dots<w(e_m)$}
\State{Initialize matrix $A$ with all zeros}
\State{Initialize matrix $D$ with all $+\infty$}
\For{$i \gets 1,2,\dots,n$} 
\State{$A_{i,i} \gets 1$}
\EndFor
\For{$k \gets 1,2,\dots,m$} 
\State{Let $i,j$ be the two endpoints of edge $e_k$}
\For{$s  \gets 1,2,\dots,n$}
\If{$A_{s,i}=1$ and $A_{s,j}=0$}
\State{$A_{s,j} \gets 1$}
\State{$D_{s,j} \gets  w(e_k)$}
\EndIf
\If{$A_{s,j}=1$ and $A_{s,i}=0$}
\State{$A_{s,i} \gets 1$}
\State{$D_{s,i} \gets  w(e_k)$}
\EndIf
\EndFor
\EndFor
\end{algorithmic}
\label{basicundir}
\end{algorithm}

\begin{lemma}
\label{unlemma}
When the $k$-th iteration of the for-loop at Line 6 is finished, for every pair of vertices $i,j$, $A_{i,j}=1$ if an only if there exists a non-decreasing path $(i,j)$ consisting only of edges in $E_k = \{e_1,\dots,e_k\}$. 
\end{lemma}
\begin{proof}
 For $k=0$, the for-loop has not started, and $A_{i,j}=1$ if and only if $i=j$. The lemma trivially holds since a path $(i,j)$ with edges in $E_0=\emptyset$ must satisfy $i=j$.
 
 For $1\le k \le n$, suppose the lemma holds for $k-1$. If $A_{i,j}=1$ before the $k$-th iteration, it should still be $1$ after the $k$-th iteration since $E_{k-1}\subset E_k$. Our algorithm never updates $A_{i,j}$  from 1 to 0.
 
 If $A_{s,t}=0$ before the $k$-th iteration, it should be updated to 1 in the $k$-th iteration if and only if there exists a non-decreasing path $(s,t)$ containing $e_k$. Since $w(e_k) = \max_{e\in E_k}w(e)$, it can only appear as the last edge of a non-decreasing path. Let $e_k = (i,j)$. Then the two possibilities of this path are $s\to i \to j$ and $s \to j \to i$. Our algorithm checks each possible $s$ and can find all such new paths.
\end{proof}

\begin{theorem}
 Algorithm \ref{basicundir} correctly computes the APNP matrix $D$.
\end{theorem}
\begin{proof}
From Lemma \ref{unlemma}, if $A_{i,j}$ is updated from $0$ to $1$ in the $k$-th iteration, there exists a non-decreasing path $(i,j)$ using edges $\{e_1,\dots,e_k\}$, but it does not exist when only $\{e_1,\dots,e_{k-1}\}$ can be used. Since $w(e_1)<w(e_2)<\dots<w(e_m)$, we must have $OPT(i,j) = w(e_k)$.
\end{proof}

\subsection{Equality-tests on dynamic strings}
To speed up our algorithm, it is necessary to maintain the zero-one matrix more efficiently. We will use the data structure introduced in \cite{mehlhorn1997maintaining}, which supports  useful operations on a dynamic family of zero-one strings. A zero-one string $s$ of length $n$ can be viewed as an array of elements $s_1,s_2,\dots,s_n$, where $s_i \in \{0,1\}$ for every $i$.

\begin{theorem}[\cite{mehlhorn1997maintaining}]
There is a deterministic data structure that supports the following operations on an initially empty family of zero-one strings.
\begin{itemize}
 \item $\textsc{Equal}(s_1,s_2)$: Test if two strings $s_1,s_2$ are equal.
 \item $\textsc{Makestring}(s,c)$: Create string $s$ with a single element $c \in \{0,1\}$ .
 \item $\textsc{Concatenate}(s_1,s_2,s_3)$: Create a new string $s_3=s_1s_2$ without destroying $s_1$ and $s_2$.
 \item $\textsc{Split}(s_1,s_2,s_3,i)$: Create two new strings $s_2=a_1\dots a_i$ and $s_3=a_{i+1}\dots a_{|s_1|}$ without destroying $s_1 = a_1,\dots a_{|s_1|}$.
\end{itemize}
Let $n$ denote the total length of all strings created.  Then the $m$-th operation takes $\mathrm{poly} \log (nm)$ time.
\end{theorem}

Creating a string $s$ can be performed using $O(|s|)$ \textsc{Makestring} and \textsc{Concatenate} operations. Modifying one bit of a string can be performed using $O(1)$ \textsc{Split}, \textsc{Makestring} and \textsc{Concatenate} operations. Given two strings $s,t$, using standard binary search,  we can find the smallest $i$ (if exists) such that $s_i \neq t_i$ using $O(\log |s|)$ \textsc{Split} and \textsc{Equal} operations.

\subsection{Faster algorithm}
Using this data structure, our algorithm can be implemented as in Algorithm~\ref{alg:undirected2}. In this new algorithm, matrix $A$ is replaced by $n$ zero-one strings $B_i$, which are maintained by the data structure. $B_{i,j}$ is the $j$-th element in $B_i$. $B_{i,s}$ corresponds to $A_{s,i}$ in the previous basic algorithm, that is, $B_i$ corresponds to the $i$-th column of $A$. The correctness of this algorithm is obvious since it is equivalent to Algorithm~\ref{alg:undirected1}.

\begin{algorithm} 
\caption{Faster algorithm for undirected APNP using equality-test data structure} \label{alg:undirected2}
\begin{algorithmic}[1] 
\State{Sort the edges so that $w(e_1)<w(e_2)<\dots<w(e_m)$}
\State{Initialize matrix $D$ with all $+\infty$}
\For{$i \gets 1,2,\dots,n$} 
\State{Create a zero-one string $B_i$ of length $n$, where $B_{i,i} = 1, B_{i,j} = 0 \,(j\neq i)$}
\EndFor
\For{$k \gets 1,2,\dots,m$} 
\State{Let $i,j$ be the two endpoints of edge $e_k$}
\While{$B_i\neq B_j$} 
  \State{Find the smallest $s$ such that $B_{i,s} \neq B_{j,s}$}\Comment{Using $O(\log n)$ operations}
  \State{$B_{i,s} \gets 1$}
  \State{$B_{j,s} \gets 1$}
  \State{$D_{s,i} \gets \min\{D_{s,i}, w(e_k)\}$}
  \State{$D_{s,j} \gets \min\{D_{s,j}, w(e_k)\}$}
\EndWhile
\EndFor
\end{algorithmic}
\end{algorithm}

\begin{theorem}
The all pairs non-decreasing paths (APNP) problem on undirected graphs can be solved in $\tilde{O}(n^{2})$ time.
\end{theorem}

\begin{proof}
To analyze the time complexity, note that creating the strings $\{B_i\}$ takes $O(n^2)$ \textsc{Makestring} and \textsc{Concatenate} operations.
The for-loop at Line 5 is executed for $m=O(n^2)$ times. Every time the while-loop from Line 8 to Line 12 is executed, one element of either $B_i$ or $B_j$ is updated from 0 to 1.
Since the total length of strings $B_i$ is $n^2$, and elements are never changed from 1 back to 0, the number of updates is $O(n^2)$.
Each execution of the comparison and the while-loop takes $\mathrm{poly} \log (n)$ time by using the data structure. (Note that comparing $B_{i}$ and $B_{j}$ and finding the first different bit take $O(\log n)$ operations by binary search.) Hence, the total time of this algorithm is $O(n^2 \mathrm{poly} \log (n)) = \tilde O(n^2)$.
\end{proof}

\bibliography{lipics-v2019-sample-article}

\appendix

\section{Handling equal weights} \label{equal-weight}
In this section we discuss the generalizations of our algorithm to multiple edges and equal edge weights.
\begin{observation}
We can generalize our algorithm for directed graphs to handle multiple edges, but still with distinct edge weights. When the number of edges $m=O(n^2)$, the running time is still $\tilde{O}\left(n^{\frac{3 + \omega}{2}}\right)$
\end{observation}
The Dijkstra search is not affected by multiple edges, so the edge relaxation of low edges and high-low edges parts still works. For high-high edges, since we relax a set of edges $H'_{[x][1]}$ w.r.t. paths with smaller weights together, we can simply only keep the smallest edge when there are multiple edges in $H'_{[x][1]}$ between a pair of vertices.

This observation is helpful to handle edges of equal weights:
\begin{lemma} \label{equal}
Given a directed simple graph $G=(V,E)$, one can construct in $\tilde O(n^2)$ time a new directed graph $H=(V,E')$ (possibly having multiple edges) on the same vertex set $V$ with distinct edge weights, such that the APNP matrix of $G$ can be computed in $\tilde O(n^2)$ time given the APNP matrix of $H$ as input. The size of the new edge set $E'$ satisfies $|E'|\le 2|E|$. 

For undirected graph $G$ this statement also holds, and the constructed $H$ will be undirected as well.
\end{lemma}
\begin{proof}

For a given number $w$, let $G_w=(V,E_w)$ denote the subgraph of $G$ consisting of all edges with weight exactly $w$. For all edge weight $w$, we construct in $O(|E_w|)$ time a new graph $H_w=(V,E_w')$ where  $|E_w'|\le 2|E_w|$ and all edge weights are distinct, such that: for any $s,t\in V$, $G_w$ contains a path from $s$ to $t$ if and only if $H_w$ contains a non-decreasing path (actually strictly increasing) from $s$ to $t$. Once this is shown, the proof immediately follows by merging all $H_w$'s into a single graph $H=(V,\, \bigcup_{w}E_w')$, assuming edges from $E'_x$ all have smaller weights than those from $E'_y$ if $x<y$.

Now we describe our construction of $H_w$ in undirected and directed cases respectively.

\paragraph*{Undirected case}

For every connected component $C=\{v_1,v_2,\dots,v_k\}$ in $G_w$,  we set $v_1$ as the "assembly vertex", and construct $2(k-1)$ undirected edges $(v_1,v_2),\dots,(v_1,v_k)$, $(v_k,v_1),\dots,(v_2,v_1)$ with increasing weights. So for each pair of (non-assembly) vertices $v_i,v_j$ in the same component, there exists a strictly increasing path $v_i\rightarrow v_1 \rightarrow v_j$ in between. We can treat all connected components of $G_w$ in arbitrary order.

\paragraph*{Directed case}

Let $V_w$ be the set of vertices associated with edges in $E_w$. We contract every strongly connected components (SCC) of $G_w$  into a big node, and thus obtain a directed acyclic graph (DAG) $G'$.
Suppose $G'$ contains $n'$ big nodes (SCCs of $G$) and $m'$ edges between different big nodes. Our construction of $H_w$ consists of $2(|V_w| - n') + m'\leq 2|E_w|$ edges with weights in strictly increasing order, described as below. 

Initially, $H_w$ contains no edge. We add edges one by one into $H_w$ in increasing order of their weights. For each strongly connected component $C=\{v_1, \dots, v_k\}$ in $G_w$, we select (arbitrarily) vertex $v_1$ as an assembly vertex. The first $|V_w| - n'$ edges we add are those edges from each non-assembly vertex to the assembly vertex in its component. Namely, for each component, we add $k - 1$ edges $(v_2,v_1) ,\dots, (v_k, v_1)$ with distinct weights. 

The next $m'$ edges we add comes from the edges in the DAG $G'$. We sort big nodes in $G'$ in topological order, and process each edge $(u,v) \in G'$ by the topological order of $u$. For $(u,v) \in G'$ from component $u$ to component $v$, we add $(u',v')$ to $H_w$, where $u',v'$ are the assembly vertices of SCCs $u$ and $v$, respectively. 

The last $|V_w| - n'$ edges we add are edges from each assembly vertex to every other vertex in its component. Namely, for each component $C=\{v_1, \dots, v_k\}$, we add $k - 1$ edges, $(v_1, v_2), (v_1, v_3), \dots, (v_1, v_k)$. 

Then for every path $i \rightarrow j$ in $G_w$, there is a corresponding increasing path $i \rightarrow i' \rightarrow j' \rightarrow j$ in the constructed $H_w$,  where  $i',j'$ are the assembly vertices of the components containing $i$ and $j$, respectively. The fact that an increasing path exists between $i'$ and $j'$ follows from the basic property of topological order.
\end{proof}

\end{document}